\documentclass[10pt,journal, compsoc]{IEEEtran}

\usepackage{amsmath, amssymb}
\usepackage{amsthm}
\usepackage{units}
\usepackage{quantikz}
\usepackage{orcidlink}
\usepackage{hyperref}
\usepackage{ulem}
\usepackage[linesnumbered,ruled,vlined]{algorithm2e}
\usepackage{caption}
\usepackage{subcaption}

\newtheorem{theorem}{Theorem}
\newtheorem{lemma}{Lemma}
\newtheorem{corollary}{Corollary}

\newcommand{\G}[1]{{C}}
\newcommand{\J}[1]{{C_1}}
\newcommand{\K}[1]{{C_2}}
\newcommand{\LM}[1]{{C_3}}
\newcommand{\W}[1]{{U}}
\newcommand{\V}[1]{{U}}

\ifCLASSOPTIONcompsoc
  \usepackage[nocompress]{cite}
\else
  \usepackage{cite}
\fi

\begin{document}

\title{Decomposition of Multi-controlled Special \\ Unitary Single-Qubit Gates}

\author{
        Rafaella Vale\orcidlink{0000-0002-5193-1409},
        Thiago Melo D. Azevedo\orcidlink{0000-0002-1068-1618},
        Ismael C. S. Araújo \orcidlink{0000-0001-6240-0235},
        Israel F. Araujo\orcidlink{0000-0002-0308-8701},
        Adenilton J. da Silva\orcidlink{0000-0003-0019-7694}%
\IEEEcompsocitemizethanks{

\IEEEcompsocthanksitem R. Vale, T.M.D. Azevedo, I.C.S. Araújo and A.J. da Silva are with Centro de Informática, Universidade Federal de Pernambuco, Recife, Pernambuco, Brazil.
\IEEEcompsocthanksitem I.F Araujo is with Department of Statistics and Data Science, Yonsei University, Seoul, Republic of Korea.
}
\thanks{Manuscript received April 19, 2005; revised August 26, 2015.}}

\markboth{Submitted to an International Journal}%
{}

\IEEEtitleabstractindextext{%
\begin{abstract}
Multi-controlled unitary gates have been a subject of interest in quantum computing since its inception, and are widely used in quantum algorithms. The current state-of-the-art approach to implementing n-qubit multi-controlled gates involves the use of a quadratic number of single-qubit and CNOT gates. However, linear solutions are possible for the case where the controlled gate is a special unitary SU(2). The most widely-used decomposition of an n-qubit multi-controlled SU(2) gate requires a circuit with a number of CNOT gates proportional to 28n.
In this work, we present a new decomposition of n-qubit multi-controlled SU(2) gates that requires a circuit with a number of CNOT gates proportional to 20n, and proportional to 16n if the SU(2) gate has at least one real-valued diagonal. This new approach significantly improves the existing algorithm by reducing the number of CNOT gates and the overall circuit depth. As an application, we show the use of this decomposition for sparse quantum state preparation. Our results are further validated by demonstrating a proof of principle on a quantum device accessed through quantum cloud services.
\end{abstract}

\begin{IEEEkeywords}
Quantum Computing, Quantum Circuit Optimization, Quantum Gate Decomposition, Multi-controlled Quantum Gates.
\end{IEEEkeywords}}

\maketitle

\IEEEdisplaynontitleabstractindextext

\IEEEpeerreviewmaketitle

\IEEEraisesectionheading{\section{Introduction}\label{sec:introduction}}

\IEEEPARstart{T}{he} prospect of quantum speedup for some computational tasks, such as prime number factoring~\cite{shor1999polynomial} and unstructured search \cite{grover1997quantum} motivated research on novel quantum computing applications. However, quantum devices in the current technology, referred to as Noisy Intermediate-Scale Quantum (NISQ) devices, are constrained by the number of qubits and amount of noisy operations~\cite{preskill2018quantum}.
For this reason, using quantum error correction techniques becomes nonviable for near-term hardware due to the overhead cost in additional required qubits and the increased gate count~\cite{takagi_fundamental_2022}.

To move towards practical quantum advantage~\cite{boixo2018characterizing}, we need to overcome these limitations in quantum devices by achieving, for instance, noise reduction in quantum operations.
There are different approaches to this problem.
One possible solution is to reduce the depth and gate count of quantum circuits through quantum circuit optimization techniques~\cite{brugiere2021reducing}.
Strategies to reduce circuit depth include compilation processes~\cite{nguyen2021enabling}, the use of auxiliary qubits~\cite{he2017decompositions}, divide-and-conquer approaches~\cite{brugiere2021reducing}, and approximated quantum circuits~\cite{araujo2021approximated}.

The decomposition of multi-controlled gates into a set of elementary gates is necessary for algorithm implementation in current quantum devices. Several works contributed to the decomposition of unitary matrices into quantum circuits. Some take into consideration general unitary gates~\cite{barenco_1995,iten2016quantum}, while others tackle more specific operations, such as multi-controlled single-qubit gates \cite{adenilton2022linear}, special unitary gates \cite{barenco_1995,iten2016quantum}, $\sigma_x$ \cite{barenco_1995,iten2016quantum,saeedi2013linear,he2017decompositions,gokhale2019qutrits,biswal2019fault,balauca2022efficient}, $R_x$ \cite{tomesh2022variational}, and phase gates \cite{kim2018ancillary}. 

Ref.~\cite{adenilton2022linear} shows how to decompose an $n$-controlled single-qubit gate $U$ with $O(n^2)$ gates and linear circuit depth.
However, this decomposition requires the calculation of $\sqrt[2^n]{U}$, which generates numerical errors and does not allow the decomposition of controlled gates with any number of qubits.
An $n$-qubit multi-controlled $SU(2)$ gate requires $28n-88$ CNOT gates for even $n$ and $28n-92$ CNOT gates for odd $n$~\cite{barenco_1995, iten2016quantum}.
Here, we show how to decompose an $n$-qubit multi-controlled $SU(2)$ gate with at most $20n-38$ CNOTs. If the $n$-qubit multi-controlled $SU(2)$ has at least one real-valued diagonal, the proposed decomposition requires at most $16n-40$ CNOTs. The proposed decomposition is based on~\cite{he2017decompositions} but does not require auxiliary qubits.

The proposed method has applications in several quantum algorithms~\cite{lloyd2014quantum, harrow2009quantum, li2023quantum, guo2022quantum, schuld2016prediction, biamonte2017quantum, park2019circuit, gleinig2021efficient, low2014quantum, orus2019quantum, mozafari2021efficient, de2021classical, de2020circuit}, for instance, in quantum machine learning, unitary matrix compilation into quantum circuits and quantum state preparation. In this work, we show the impact of the proposed method in sparse quantum state initialization.

The rest of this paper has four sections.
Section~\ref{sec:related_work} describes the related works.
Section~\ref{sec:results} is the main section. We first show how to decompose a multi-controlled $SU(2)$ gate with one real diagonal with the number of CNOTs proportional to $16n$ and then use this result to decompose any multi-controlled $SU(2)$ with the number of CNOTs proportional to $20n$.
Section~\ref{sec:experiments} shows the application of the proposed decomposition to reduce the amount of CNOT gates to initialize a sparse quantum state.
Section~\ref{sec:conclusion} presents the conclusion.

\section{Related work} 
\label{sec:related_work}
An implementation of an $n$-controlled $R_x(\pi)$ gate was proposed in \cite{saeedi2013linear}.
This construction consists in chaining controlled gates so that operations with different controls and target qubits are applied in parallel.
That construction was further generalized to an $n$-controlled unitary with linear depth and $4 n^2 - 12 n + 10$ CNOTs \cite{adenilton2022linear}.
However, to calculate the angles used on the operators those methods rely heavily on computing fractions of $\pi$ such as $\pi/2^{n-1}$ for a number $n$ of control qubits or computing the $2^n$--th root of the desired operator as in $\sqrt[2^{n-1}]{U}$ in the case of \cite{adenilton2022linear}, leading to numerical errors when calculating angles in systems with many qubits.

The authors in~\cite{barenco_1995} present several gate decompositions for operations with single and multiple controls. The method for multi-controlled gates in $U(2)$ is subject to numerical errors as in~\cite{adenilton2022linear} and uses $16n^2 - 60n + 42$ CNOTs.
The gate decomposition for multi-controlled gates in $SU(2)$ has linear complexity in terms of basic operations and produces circuits with $32n - 74$ CNOTs.
In \cite{iten2016quantum} the authors lay out a construction that further optimizes the decomposition of $SU(2)$ gates in \cite{barenco_1995} requiring $28n-88$ CNOT gates for even $n$ and $28n-92$ CNOT gates for odd $n$.
This optimization is accomplished by utilizing an approximate version of the Toffoli gate with a phase change on state $\ket{010}$.

The construction shown in \cite{he2017decompositions}
is the basis for the decomposition scheme proposed in this paper, whereas the works of \cite{barenco_1995} and \cite{iten2016quantum} are used for comparison. Thus, they will be summarized in Section~\ref{sec:barenco_iten} and Section~\ref{sec:he}.

\subsection{Multi-controlled special unitary gates}
\label{sec:barenco_iten}
In order to build an $(n - 1)$-controlled version of a gate $\W{} \in SU(2)$ as \cite{barenco_1995}, we must find a gate decomposition as illustrated in Fig.~\ref{fig:ncontrolled-w}.
Then we can use the $(n - 1)$-th qubit as an auxiliary for decomposing the two multi-controlled $\sigma_x$ gates more efficiently.
Such scheme is derived from the fact that the gate $\W{}$ can be decomposed into three gates $A, B$ and $C$ such that $ABC = I$ and $\W{} = A \sigma_x B \sigma_x C$, where $A, B$ and $C \in SU(2)$ if and only if $\W{} \in SU(2)$ \cite[Lemma 7.9]{barenco_1995}.

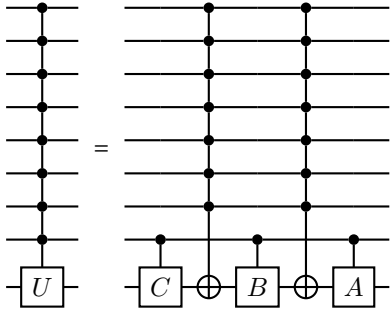
\begin{figure}[htb]
    \centering
    \begin{quantikz}[column sep=0.2cm, row sep=0.3cm]
     & \ctrl{1} & \qw \\
     & \ctrl{1} & \qw \\
     & \ctrl{1} & \qw \\
     & \ctrl{1} & \qw \\
     & \ctrl{1} & \qw \\
     & \ctrl{1} & \qw \\
     & \ctrl{1} & \qw \\
     & \ctrl{1} & \qw \\
     & \gate{\W{}} & \qw 
    \end{quantikz}
    ~=\begin{quantikz}[column sep=0.2cm, row sep=0.3cm]
     &   \qw    &\ctrl{1}  &  \qw      &\ctrl{1} & \qw      & \qw\\
     &   \qw    &\ctrl{1}  &  \qw      &\ctrl{1} & \qw      & \qw\\
     &   \qw    &\ctrl{1}  &  \qw      &\ctrl{1} & \qw      & \qw\\
     &   \qw    &\ctrl{1}  &  \qw      &\ctrl{1} & \qw      & \qw\\
     &   \qw    &\ctrl{1}  &  \qw      &\ctrl{1} & \qw      & \qw\\
     &   \qw    &\ctrl{1}  &  \qw      &\ctrl{1} & \qw      & \qw\\
     &   \qw    &\ctrl{2}  &  \qw      &\ctrl{2} & \qw      & \qw\\
     & \ctrl{1} &  \qw     &  \ctrl{1} &\qw      & \ctrl{1} & \qw \\
     & \gate{C} & \targ{}  & \gate{B}  & \targ{} & \gate{A} & \qw
    \end{quantikz}
    
    \caption{Decomposition of an $8$-controlled $\W{}$ gate, where $\W{}, A, B, C \in SU(2)$.}
    \label{fig:ncontrolled-w}
\end{figure}

Since we know that $A$, $B$ and $C$ are also special unitary gates, then their controlled versions can be decomposed according to a similar reasoning used for the $(n -  1)$-controlled $\W{}$ (Fig.~\ref{fig:controlled-g}).
But this construction would only need to take into account a single-control gate \cite[Lemma 5.1]{barenco_1995}.
With those constructions introduced, we can focus on the $(n - 2)$-controlled $\sigma_x$ gates with the extra $(n - 1)$-th qubit as an auxiliary. 

\begin{figure}[htb]
    \centering
    \begin{quantikz}[column sep=0.2cm]
        &\ctrl{1}& \qw \\
        &\gate{\G{}}& \qw
    \end{quantikz}
    =
    \begin{quantikz}[column sep=0.2cm]
        &\qw     & \ctrl{1} & \qw      & \ctrl{1} & \qw    &\qw\\
        &\gate{\J{}}& \targ{}& \gate{\K{}} & \targ{}&\gate{\LM{}}&\qw
    \end{quantikz}
    \caption{Single-control $C$ gate, where $\G{}, \J{}, \K{}, \LM{} \in SU(2)$.}
    \label{fig:controlled-g}
\end{figure}
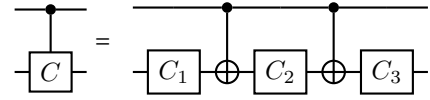

On an $n$-qubit system where $n \geq 5$ with one auxiliary qubit, an $(n-2)$-controlled $\sigma_x$ can be decomposed into two pairs of $m_1$-controlled $\sigma_x$ and $m_2$-controlled $\sigma_x$ gates, where $m_2 = \lceil n/2 \rceil$ and $m_1 = n - m_2 - 1$ \cite[Lemma 7.3]{barenco_1995}.
Each gate is further decomposed into a chain of Toffoli gates, with the unused qubits as auxiliary qubits as illustrated in Fig.~\ref{fig:toffoli-chain-decomposition}.
The first step is dedicated to flipping the target qubit and the second to reverting the auxiliary qubits.
The auxiliary qubit can be dirty since they are reset to their original values.

\begin{figure}[htb]
    \centering
    \begin{quantikz}[column sep=0.2cm, row sep=0.2cm]
        & \ctrl{1} & \ghost{} \qw \\
        & \ctrl{1} & \ghost{} \qw \\
        & \ctrl{1} & \ghost{} \qw \\
        & \ctrl{3} & \ghost{} \qw \\
        & \qw      & \ghost{} \qw \\
        & \qw      & \ghost{} \qw \\
        & \targ{}  & \ghost{} \qw
    \end{quantikz}
    ~=\begin{quantikz}[column sep=0.3cm, row sep=0.2cm]
        & \ghost{}\qw & \qw\gategroup[7, steps=5, style={dashed, rounded corners}, background]{{flip target}}                            &\qw     &\ctrl{1}&\qw     &\qw     &\qw &\qw\gategroup[7, steps=3, style={dashed, rounded corners}, background]{{revert auxiliary}}     &\ctrl{1}&\qw     & \qw & \qw \\
        & \ghost{} \qw & \qw      &\qw     &\ctrl{3}&\qw     &\qw     &\qw &\qw     &\ctrl{3}&\qw     & \qw & \qw\\
        & \ghost{} \qw & \qw      &\ctrl{2}&\qw     &\ctrl{2}&\qw     &\qw &\ctrl{2}&\qw     &\ctrl{2} & \qw & \qw \\
        & \ghost{} \qw & \ctrl{2} &\qw     &\qw     &\qw     &\ctrl{2}&\qw &\qw     &\qw     &\qw     & \qw & \qw \\
        & \ghost{} \qw & \qw      &\ctrl{1}&\targ{} &\ctrl{1}&\qw     &\qw &\ctrl{1}&\targ{} &\ctrl{1} & \qw & \qw \\
        & \ghost{} \qw & \ctrl{1} &\targ{} &\qw     &\targ{} &\ctrl{1}&\qw &\targ{} &\qw     &\targ{} & \qw & \qw \\
        & \ghost{} \qw & \targ{}  &\qw     &\qw     &\qw     &\targ{} &\qw &\qw     &\qw     &\qw     & \qw & \qw
    \end{quantikz}
    \caption{Decomposition of a $4$-controlled $\sigma_x$ in a $7$-qubit system with two dirty auxiliary qubits \cite[Lemma 7.2]{barenco_1995}. In the optimizations proposed in \cite[Lemma 8]{iten2016quantum}, all Toffoli gates that target the auxiliary qubits are approximated.}
    \label{fig:toffoli-chain-decomposition}
\end{figure}
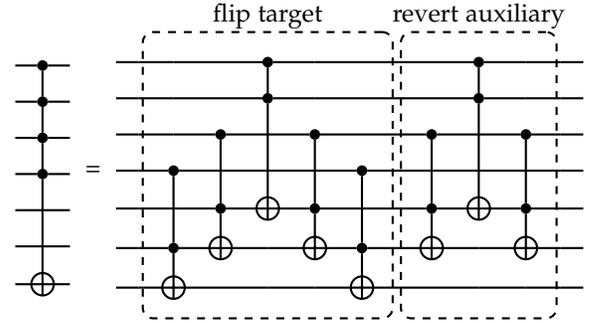

The construction of the chain of Toffoli gates can be further optimized~\cite{barenco_1995, iten2016quantum}.
By choosing an appropriate rotation gate (for example, $R_y(\pi/4)$), an approximated Toffoli operation can be built with just three CNOTs.
However, the local phase for some input states will differ from the action of an actual Toffoli gate, which makes this optimization equivalent to a Toffoli up to some diagonal gate.
In the example of Fig.~\ref{fig:toffoli-optimised}, this is similar to adding a diagonal gate $\Delta$ that performs the mapping $\ket{010} \mapsto -\ket{010}$ after (or before) the Toffoli gate.
This is a useful optimization when the circuit can be designed to match gates of this type in a way that forces gate canceling to further reduce the depth and gate count of the circuit.

\begin{figure}[htb]
    \centering
    \begin{quantikz}[column sep=0.1cm, row sep=0.2cm]
        & \ghost{x} \qw           & \ctrl{1} & \gate[3]{\Delta} & \qw \\
        & \ghost{x} \qw           & \ctrl{1} &                  & \qw \\
        & \ghost{R_y^\dagger} \qw & \targ{}  &                  & \qw
    \end{quantikz}
    ~=\begin{quantikz}[column sep=0.3cm, row sep=0.2cm]
    & \ghost{x} \qw      & \ctrl{2} & \qw                & \qw      & \qw        & \ctrl{2} & \qw        & \qw \\
    & \ghost{x} \qw      & \qw      & \qw                & \ctrl{1} & \qw        & \qw      & \qw        & \qw \\
    & \gate{R_y^\dagger} & \targ{}  & \gate{R_y^\dagger} & \targ{}  & \gate{R_y} & \targ{}  & \gate{R_y} & \qw
    \end{quantikz}
    \caption{Optimized version of a Toffoli gate up to a diagonal gate~\cite{barenco_1995, iten2016quantum}. In this example, the angles of all $R_y$ operators are equal to $\pi/4$.}
    \label{fig:toffoli-optimised}
\end{figure}
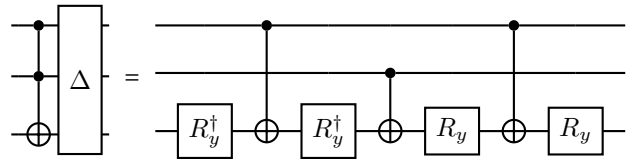

The approximate version of the Toffoli gate shown in Fig.~\ref{fig:toffoli-optimised} can be used to replace each auxiliary-targeting Toffoli gate on the chain illustrated in Fig.~\ref{fig:toffoli-chain-decomposition},
where each gate pairs up with another acting on the same qubits.
Doing this, we obtain a sequence of operations that eliminate the extra phases and cancel nearly half of the gates in the circuit shown in Fig.~\ref{fig:toffoli-optimised}.
This optimization step is demonstrated in \cite[Lemma 8]{iten2016quantum}.
This consequently enables us to construct an efficient version of the $(n-2)$-controlled $\sigma_x$ gate with one auxiliary qubit.

As mentioned before, two pairs of $m_1$-controlled $\sigma_x$ gates in the top $m_1$ control qubits and $m_2$-controlled $\sigma_x$ gates in the bottom $m_2$ control qubits, such as the one shown in Fig.~\ref{fig:toffoli-chain-decomposition}, implementing the optimizations discussed, are used to compose the $(n-2)$-controlled $\sigma_x$ circuit.
Fig.~\ref{fig:ncontrolled-x-one-aux} illustrates how \cite[Lemma 8]{iten2016quantum} is used to construct this scheme stated in \cite[Lemma 9]{iten2016quantum}.
One thing to note is that the gates controlled by the top qubits act only on the auxiliary qubits and not on the target of the main operation being decomposed.
Due to this, all Toffoli gates in the top gates can be substituted by their approximate counterparts, since the auxiliary qubits are reset at the end of the operation.
This particular step is also noted in \cite[Corollary 7.4]{barenco_1995}.
This set of optimizations for the circuit seen in Fig.~\ref{fig:ncontrolled-x-one-aux} has a cost of $16n-40$ CNOTs.

\begin{figure}[htb]
    \centering
    \begin{quantikz}[column sep=0.2cm, row sep=0.3cm]
     & \ctrl{1}     & \qw \\
     & \ctrl{1}     & \qw \\
     & \ctrl{1}     & \qw \\
     & \ctrl{1}     & \qw \\
     & \ctrl{1}     & \qw \\
     & \ctrl{1}     & \qw \\
     & \ctrl{2}     & \qw \\
     & \ghost{} \qw & \qw \\
     & \targ{}      & \qw
    \end{quantikz}
    ~=\begin{quantikz}[column sep=0.2cm, row sep=0.3cm]
     & \ctrl{1}  & \qw      & \ctrl{1} & \qw      & \qw \\
     & \ctrl{1}  & \qw      & \ctrl{1} & \qw      & \qw \\
     & \ctrl{1}  & \qw      & \ctrl{1} & \qw      & \qw \\
     & \ctrl{4}  & \qw      & \ctrl{4} & \qw      & \qw \\
     & \qw       & \ctrl{1} & \qw      & \ctrl{1} & \qw \\
     & \qw       & \ctrl{1} & \qw      & \ctrl{1} & \qw \\
     & \qw       & \ctrl{1} & \qw      & \ctrl{1} & \qw \\
     & \targ{}   & \ctrl{1} & \targ{}  & \ctrl{1} & \qw \\
     & \qw       & \targ{}  & \qw      & \targ{}  & \qw
    \end{quantikz}
    
    \caption{Decomposition of a $7$-controlled $\sigma_x$ gate with one free qubit, which can be dirty, that can be used as an auxiliary. If the four controlled $\sigma_x$ gates follow the optimized decomposition scheme of \cite[Lemma 8]{iten2016quantum} and use only approximate Toffoli gates when the bottom qubit is not targeted, the total number of CNOTs in the circuit is at most $16n-40$.}
    \label{fig:ncontrolled-x-one-aux}
\end{figure}
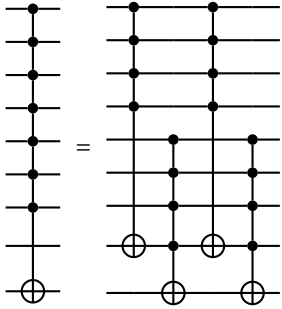

So far, these adjustments enable the decomposition of an $(n - 1)$-controlled $\W{} \in SU(2)$ with $32n-74$ CNOTs with two $(n-2)$-controlled $\sigma_x$ and the three controlled $SU(2)$ gates.
The final improvement proposed in \cite[Theorem 5]{iten2016quantum} uses the fact that the circuit in Fig.~\ref{fig:ncontrolled-x-one-aux} can be reversed without affecting its action. 
The second occurrence of this gate in Fig.~\ref{fig:ncontrolled-w} can then mirror the first.
This allows the canceling of the auxiliary-resetting part shown in Fig.~\ref{fig:toffoli-chain-decomposition} for the two bottom $m_2$-controlled $\sigma_x$ gates surrounding the controlled $B$ gate.
The bottom $m_2$-controlled $\sigma_x$ gate is the last gate represented in Fig.~\ref{fig:ncontrolled-x-one-aux}.
With these optimizations taken into account, an $(n-1)$-controlled $\W{} \in SU(2)$ can be built, with an upper bound of $28n-88$ CNOTs when $n$ is even, and $28n-92$ CNOTs when $n$ is odd.

The decomposition displayed in Fig. \ref{fig:ncontrolled-w} can be further reduced for a special case of $SU(2)$ operators.
Given an operator $\W{} \in SU(2)$ defined as 
\begin{align}    
\label{eq:special-case-w}
    \W{} &= R_z(\beta)R_y(\theta)R_z(\beta) \\
      &= \begin{pmatrix}
            e^{-i\beta}\cos{\frac{\theta}{2}}  & -\sin{\frac{\theta}{2}} \\
            \sin{\frac{\theta}{2}}            & e^{i\beta}\cos{\frac{\theta}{2}}
        \end{pmatrix},
\end{align}
its controlled version can be decomposed into two single-qubit gates $A, B \in SU(2)$ and two CNOTs, such that $AB = I$ and $\W{} = A \sigma_x B \sigma_x$.
This decomposition can be achieved by making $A = R_z(\beta)R_y(\theta/2)$ and $B=A^\dag$ \cite[Lemma 5.4]{barenco_1995}).
Applying this special case to the decomposition of Fig. \ref{fig:ncontrolled-w}, one would just need to consider the controlled $C$ gate as the identity.

\subsection{\texorpdfstring{Multi-controlled $\sigma_x$ gates with one auxiliary qubit}{Multi controlled X gates with one auxiliary qubit}} \label{sec:he} 

In \cite{he2017decompositions} the authors have proposed a gate decomposition of an $n$-controlled $\sigma_x$ operator that uses a single dirty auxiliary qubit, illustrated in Fig. \ref{fig:n-CNOT}.
In an $(n + 2)$-qubit system, with $n \ge 3$ plus the auxiliary qubit,
the procedure involves splitting the control qubits into two groups of $k_1$ and $k_2$ qubits, where $k_1 + k_2 = n$.
Then, both control qubit groups are used to manipulate the phase of the auxiliary qubit by using the $\sigma_x$ gate and the phase gate $S$, which can also be seen as flipping the phase of the target qubit. 

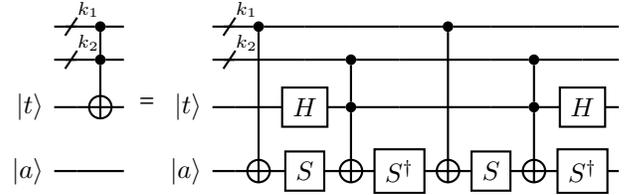
\begin{figure}[htb]
    \centering
    \begin{quantikz}[column sep=0.15cm, row sep=0.3cm] 
                     & \qw & \qw \qwbundle{k_1} & \ctrl{1} & \qw \\
                     & \qw & \qw \qwbundle{k_2} & \ctrl{1} & \qw \\
    \lstick{\ket{t}} & \qw & \qw                & \targ{}  & \ghost{H} \qw \\
    \lstick{\ket{a}} & \qw & \qw                & \qw      & \ghost{S^\dagger} \qw 
    \end{quantikz}
    ~=\begin{quantikz}[column sep=0.15cm, row sep=0.3cm]
                    & \qw & \qw \qwbundle{k_1}&\ctrl{3}&\qw     &\qw     &\qw          &\ctrl{3}&\qw     &\qw     &\qw          &\qw\\
                    & \qw & \qw \qwbundle{k_2}&\qw     &\qw     &\ctrl{1}&\qw          &\qw     &\qw     &\ctrl{1}&\qw          &\qw\\
    \lstick{\ket{t}}& \qw & \qw               &\qw     &\gate{H}&\ctrl{1}&\qw          &\qw     &\qw     &\ctrl{1}&\gate{H}     &\qw \\
    \lstick{\ket{a}}& \qw & \qw               &\targ{} &\gate{S}&\targ{} &\gate{S^\dag}&\targ{} &\gate{S}&\targ{} &\gate{S^\dag}&\qw
    \end{quantikz}
    \caption{Decomposition scheme of an $n$-controlled $\sigma_x$ operator with a single auxiliary qubit $\ket{a}$ and target qubit $\ket{t}$.}
    \label{fig:n-CNOT}
\end{figure}

The operation has no effect on the target qubit when all control qubits are 0. The action of the circuit satisfies the following equation on the auxiliary qubit when at least one of the $k_1$ controls is 0 and all $k_2$ controls are 1:
\begin{equation} \label{eq:circuit_identity_1}
    S^\dagger \sigma_x S S^\dagger \sigma_x S = I,
\end{equation}
and the following equation when the $k_1$ controls are 1 and at least one of the $k_2$ controls is 0:
\begin{equation} \label{eq:circuit_identity_2}
    S S^\dagger \sigma_x S S^\dagger \sigma_x = I.
\end{equation}
What remains to ensure the action of $\sigma_x$ on the target qubit is to flip the phase of the auxiliary qubit with the sequence of transformations
\begin{equation}
    S^\dagger \sigma_x S \sigma_x S^\dagger \sigma_x S \sigma_x = -I
\end{equation}
when all controls are 1, corresponding to a controlled $\sigma_z$ on the target, and to apply the Hadamard operators on both ends of the circuit, since $H \sigma_z H = \sigma_x$.

\section{Decomposition of multi-controlled single-qubit $SU(2)$ gates}\label{sec:results}

Different schemes can be used to decompose multi-controlled $SU(2)$ operators \cite{barenco_1995, iten2016quantum}. 
The construction from \cite[Lemma 7.9]{barenco_1995} results in linear depth and a linear number of gates. The previous optimal number of CNOTs is obtained using the optimizations of \cite[Theorem 5]{iten2016quantum}, with the total number of CNOTs equal to $28n-88$ and $28n-92$ for even and odd numbers of qubits, respectively. In this section, we present a decomposition scheme and its variations for different types of $SU(2)$ operators.

\subsection{Multi-controlled SU(2) gates with real-valued diagonal} 
Here, we modify the decomposition scheme of \cite{he2017decompositions} for multi-controlled $\sigma_x$ with one auxiliary qubit to multi-controlled $SU(2)$ gates with at least one real-valued diagonal. The quantum circuit shown in Fig.~\ref{fig:he_adapted} is based on that scheme and provides the basic structure for our result. Unlike in \cite{he2017decompositions}, there is no auxiliary qubit, and all actions occur on the target qubit.

As in~\cite{he2017decompositions}, the circuit requires every consecutive single-qubit gate pair to be inverse of each other to satisfy Equations~\eqref{eq:circuit_identity_1} and~\eqref{eq:circuit_identity_2}, with the difference that the gate represented by $A$ in Fig.~\ref{fig:he_adapted} is not restricted to just one particular gate. If all control qubits are 1, then the action on the target is
\begin{equation}\label{eq:he-mod-W}
    A^\dagger \sigma_x A \sigma_x A^\dagger \sigma_x A \sigma_x = \V{}.
\end{equation}
For simplicity, we assume the $A$ gates in the circuit are $SU(2)$ gates, so they have the form
\begin{equation} \label{eq:su2_form}
    A = \begin{pmatrix}
        \alpha & -\beta^*\\
        \beta & \alpha^*
    \end{pmatrix}.
\end{equation}
With that said, the quantum circuit still follows the same design as in~\cite{he2017decompositions}, but with some gate $A \in SU(2)$ replacing the phase gate $S$ and without an auxiliary qubit.

\begin{figure}[ht]
    \centering
    \begin{quantikz}[column sep=0.1cm, row sep=0.3cm]
    & \qw & \qwbundle{k_1} & \ctrl{1} & \qw& \qw \\
    & \qw & \qwbundle{k_2} & \ctrl{1} & \qw & \qw \\
    & \qw & \qw & \gate{\V{}} & \qw & \qw
    \end{quantikz}
    ~=\begin{quantikz}[column sep=0.1cm, row sep=0.3cm]
    & \qw & \qwbundle{k_1} & \qw & \ctrl{2}&\qw     &\qw     &\qw             &\ctrl{2} &\qw     &\qw     & \qw & \qw & \qw \\
    & \qw & \qwbundle{k_2} & \qw &\qw     &\qw     &\ctrl{1}&\qw             &\qw      &\qw     &\ctrl{1} & \qw & \qw & \qw\\
    & \qw & \qw           & \qw &\targ{} &\gate{A}&\targ{} &\gate{A^\dagger}&\targ{}  &\gate{A}&\targ{} &\gate{A^\dagger} & \qw & \qw
    \end{quantikz}
    \caption{Decomposition scheme based on \cite{he2017decompositions} for multi-controlled $SU(2)$ gates, where $\V{}, A \in SU(2)$, and $\V{}$ has at least one real-valued diagonal.}
    \label{fig:he_adapted}
\end{figure}

The first step of the decomposition of an $n$-qubit multi-controlled $\V{}$ with $k$ controls consists in dividing the control register roughly in half by choosing $k_1 = \lceil k/2 \rceil$ and $k_2 = \lfloor k/2 \rfloor$ as the size of each new register. Since they both have available qubits from the other register to use as dirty auxiliary qubits, the $k_1$-controlled and $k_2$-controlled $\sigma_x$ gates are implemented following \cite[Lemma 8]{iten2016quantum}, which improves the gate count of \cite[Lemma 7.2]{barenco_1995}.

\begin{lemma} \label{lemma-mod-he}
    An operator given by the circuit in Fig.~\ref{fig:he_adapted} whose action on the target qubit is given by Equation~\eqref{eq:he-mod-W} and where $A$ is of the form of Equation~\eqref{eq:su2_form} generates an $n$-controlled $SU(2)$ gate with the restriction that its off-diagonal is real-valued.
\end{lemma}

\begin{proof}

When all the control qubits are active, the circuit generates the matrix $A^\dagger \sigma_x A \sigma_x A^\dagger \sigma_x A \sigma_x = (A^\dagger \sigma_x A \sigma_x)^2$. It can be shown that
\begin{equation}\label{eq::aux_w}
    A^\dagger \sigma_x A \sigma_x = \begin{pmatrix}
        \omega_1^* & \omega_2 \\
        -\omega_2 & \omega_1
    \end{pmatrix},
\end{equation}
where $\omega_1 = \alpha^2 - \beta^2$ and $\omega_2 = 2 \operatorname{Re}(\alpha^* \beta)$, with $\omega_1$ being a complex number and $\omega_2$ being a real number. Since the determinant of the product $A^\dagger \sigma_x A \sigma_x$ equals 1, we can see that this is an $SU(2)$ matrix with real elements in its off-diagonal. Then
\begin{align}
    (A^\dagger \sigma_x A \sigma_x)^2 &= \begin{pmatrix}
        (\omega_1^{*})^2 - \omega_2^2 & 2 \operatorname{Re}(\omega_1) \omega_2 \\
        -2 \operatorname{Re}(\omega_1) \omega_2 & \omega_1^2 - \omega_2^2
    \end{pmatrix} \\
    &= \begin{pmatrix}
        z^* & x \\
        -x & z
        \end{pmatrix}=\V{},
\end{align}

with $z$ being a complex number and $x$ being a real number. This shows that this modification of the decomposition scheme of \cite{he2017decompositions} gives us $SU(2)$ matrices with real elements in the off-diagonal. \end{proof}

\begin{theorem}\label{off-diag-volta}
Every $n$-controlled $SU(2)$ gate whose matrix has real elements in its off-diagonal can be generated by the circuit of the operator described in Lemma~\ref{lemma-mod-he}.
\end{theorem}

\begin{proof}

Using Lemma~\ref{lemma-mod-he}:
\begin{equation} \label{eq:matrix_V}
    \begin{pmatrix}
        z^* & x \\
        -x & z
    \end{pmatrix} = \begin{pmatrix}
            (\omega_1^*)^2 - \omega_2^2 & 2 \operatorname{Re}(\omega_1) \omega_2 \\
        -2 \operatorname{Re}(\omega_1)\omega_2 & \omega_1^2 - \omega_2^2
    \end{pmatrix}
\end{equation}

Noticing that the matrices on both sides have the determinant equal to 1, we can choose the following positive solutions:
 \begin{align*}
     \omega_1 &= \sqrt{\frac{\operatorname{Re}(z)+1}{2}} + i \frac{\operatorname{Im}(z)}{\sqrt{2(\operatorname{Re}(z)+1)}} \\
      \omega_2 &= \frac{x}{\sqrt{2(\operatorname{Re}(z)+1)}}
 \end{align*}
And, as we know from Lemma~\ref{lemma-mod-he}:
\begin{align*}
    \omega_1 &= \alpha^2 - \beta^2 \\
    \omega_2 &= 2 \operatorname{Re}(\alpha^* \beta)
\end{align*}
We can make a choice for $\beta$ to be a real number. And since $\left\Vert \alpha \right\Vert^2 + \left\Vert \beta \right\Vert^2 = 1$, we can find the solutions:
\begin{multline} \label{eq:alpha}
    \alpha = \sqrt{\frac{\sqrt{\frac{\operatorname{Re}(z)+1}{2}} + 1}{2}} \\
    +i \frac{\operatorname{Im}(z)}{2 \sqrt{(\operatorname{Re}(z)+1) \left( \sqrt{\frac{\operatorname{Re}(z)+1}{2}} + 1 \right) }}
\end{multline}
\begin{equation} \label{eq:beta}
    \beta = \frac{x}{2 \sqrt{(\operatorname{Re}(z)+1) \left( \sqrt{\frac{\operatorname{Re}(z)+1}{2}} + 1 \right)}} \qquad \qquad \quad
\end{equation}

 So, we have proved that every $SU(2)$ matrix in the form of Equation~\eqref{eq:matrix_V} can be generated by the proposed modified circuit.
\end{proof}

A more detailed description of the steps taken in Theorem~\ref{off-diag-volta} can be found in Appendix~\ref{apd:theorem1_details}.

With a small modification, the proposed circuit of Fig.~\ref{fig:he_adapted} can also generate $n$-controlled $SU(2)$ gates with real-valued elements in their main diagonal, while the off-diagonal could be complex.

\begin{lemma} \label{main-diag-ida}
   A modification of Lemma~\ref{lemma-mod-he} can be made such that when all the control qubits are active, it generates an $n$-controlled $SU(2)$ gate with the restriction that the main diagonal contains real elements.
\end{lemma}

\begin{proof}
By modifying the target qubit basis with a pair of Hadamard gates, the proposed circuit generates an $SU(2)$ matrix with a real main diagonal, as shown in Equation~\eqref{eq:change_basis_modified_he}.
\begin{equation} \label{eq:change_basis_modified_he}
\begin{split}
    H\begin{pmatrix}
        z^* & x \\
        -x & z
    \end{pmatrix}H&=
    \begin{pmatrix}
        \operatorname{Re}(z) & -x-i\operatorname{Im}(z) \\
        x-i\operatorname{Im}(z) & \operatorname{Re}(z)
    \end{pmatrix} \\
    &=
    \begin{pmatrix}
        x' & -z' \\
        z'^* & x'
    \end{pmatrix}
\end{split}
\end{equation}
Since $H$ is its own inverse, this procedure preserves the identities of Equations~\eqref{eq:circuit_identity_1} and~\eqref{eq:circuit_identity_2}.
\end{proof}

\begin{theorem} \label{main-diag-volta}
Every $n$-controlled $SU(2)$ gate whose matrix has real elements in its main diagonal can be generated by the circuit of the operator described in Lemma~\ref{main-diag-ida}.
\end{theorem}

\begin{proof}
With the change of basis from Lemma~\ref{main-diag-ida}, Theorem~\ref{off-diag-volta} can be adapted such that the real part of $z$ encodes the real-valued main diagonal, and both $x$ and the imaginary part of $z$ encode the complex off-diagonal.
\begin{equation} \label{eq:change_basis}
\begin{split}
    x &= \operatorname{Re}(z') \\
    z &= x' + i\operatorname{Im}(z')
\end{split}
\end{equation}
The construction of gate $A$ is modified using Equation~\eqref{eq:change_basis} to change Equations~\eqref{eq:alpha} and~\eqref{eq:beta}, so that every $SU(2)$ matrix in the form of Equation~\eqref{eq:change_basis_modified_he} can be generated.
\end{proof}

The proposed circuit can also generate $n$-controlled $SO(2)$ gates, which leads to Corollary~\ref{co:so2}.
\begin{corollary} \label{co:so2}
    Every $SO(2)$ matrix can be generated by Theorem~\ref{off-diag-volta} and Theorem~\ref{main-diag-volta}.
\end{corollary}

\subsubsection{Complexity}
As noted before, in between the single-qubit gates, the circuit implements the decomposition for the multi-controlled $\sigma_x$ gate with multiple auxiliary qubits present in \cite[Lemma 7.2]{barenco_1995} and includes the optimizations described in \cite[Corollary 7.4]{barenco_1995} and \cite[Lemma 8]{iten2016quantum}. Then, for each multi-controlled $\sigma_x$ with at least five qubits, that is, with the number of controls $k \geq 3$ and at least $k-2$ auxiliary qubits, at most $8k-6$ CNOTs are needed. This means that after the subdivision into two control registers of sizes $k_1$ and $k_2$, the maximum total number of CNOTs is $2 (8 k_1 - 6) + 2 (8 k_2 - 6) = 16 (k_1 + k_2) - 24$. Given that $k_1 + k_2 = n - 1$, where $n$ is the number of qubits, we can state the following theorem.

\begin{theorem}
    \label{th:cnot-count-real-diag}
    The quantum circuit shown in Fig. \ref{fig:he_adapted} can be implemented as an $n$-qubit circuit, where $n \geq 3$, with at most $16n-40$ CNOTs.
\end{theorem}

\subsubsection{\texorpdfstring{Application to multi-controlled $R_x$, $R_y$ and $R_z$ gates}{Application to multi-controlled Rx, Ry and Rz gates}}

From what has been shown, the operators $R_x$, $R_y$ and $R_z$ can be decomposed using the proposed decomposition scheme, which Corollary \ref{th:r_gates} formally states.

\begin{corollary}
    \label{th:r_gates}
    The multi-controlled versions of the rotation operator gates $R_y$ and $R_z$ can be generated by the circuit of the operator described in Lemma~\ref{lemma-mod-he}, whereas the multi-controlled version of $R_x$ can be generated by the circuit of the operator described in Lemma~\ref{main-diag-ida}.
\end{corollary}

\begin{proof}
The proof follows from the application of the procedure elaborated in Theorem~\ref{off-diag-volta} to generate $R_y$ and $R_z$ and Lemma \ref{main-diag-ida} to generate $R_x$.
Alternatively, we can demonstrate that $R_x$, $R_y$ and $R_z$ can be generated by replacing the gate $A$ of Equation \eqref{eq:su2_form} with different gates, as follows:
\begin{equation}
\begin{split}
    R_x(\theta) &= H (R_z(-\theta/4) \sigma_x R_z(\theta/4) \sigma_x)^{2} H, \\
    R_y(\theta) &= (R_y(\theta/4) \sigma_x R_y(-\theta/4) \sigma_x)^{2}, \\
    R_z(\theta) &= (R_z(\theta/4) \sigma_x R_z(-\theta/4) \sigma_x)^{2}
\end{split}
\end{equation}
\end{proof}

In Fig. \eqref{fig:mcrz_cx_count}, the impact of the proposed decomposition scheme on the number of CNOTs in multi-controlled $R_z$ is shown in contrast with \cite[Lemma 7.9]{barenco_1995} with the optimizations of \cite[Theorem 5]{iten2016quantum} as described in Section~\ref{sec:barenco_iten}. The results were obtained from Qiskit's \cite{qiskit} transpilation routine with no additional optimizations and assume complete qubit connectivity. The basis gate set specified consisted of single qubit and CNOT gates.

\begin{figure}
    \centering
    \includegraphics[width=\linewidth]{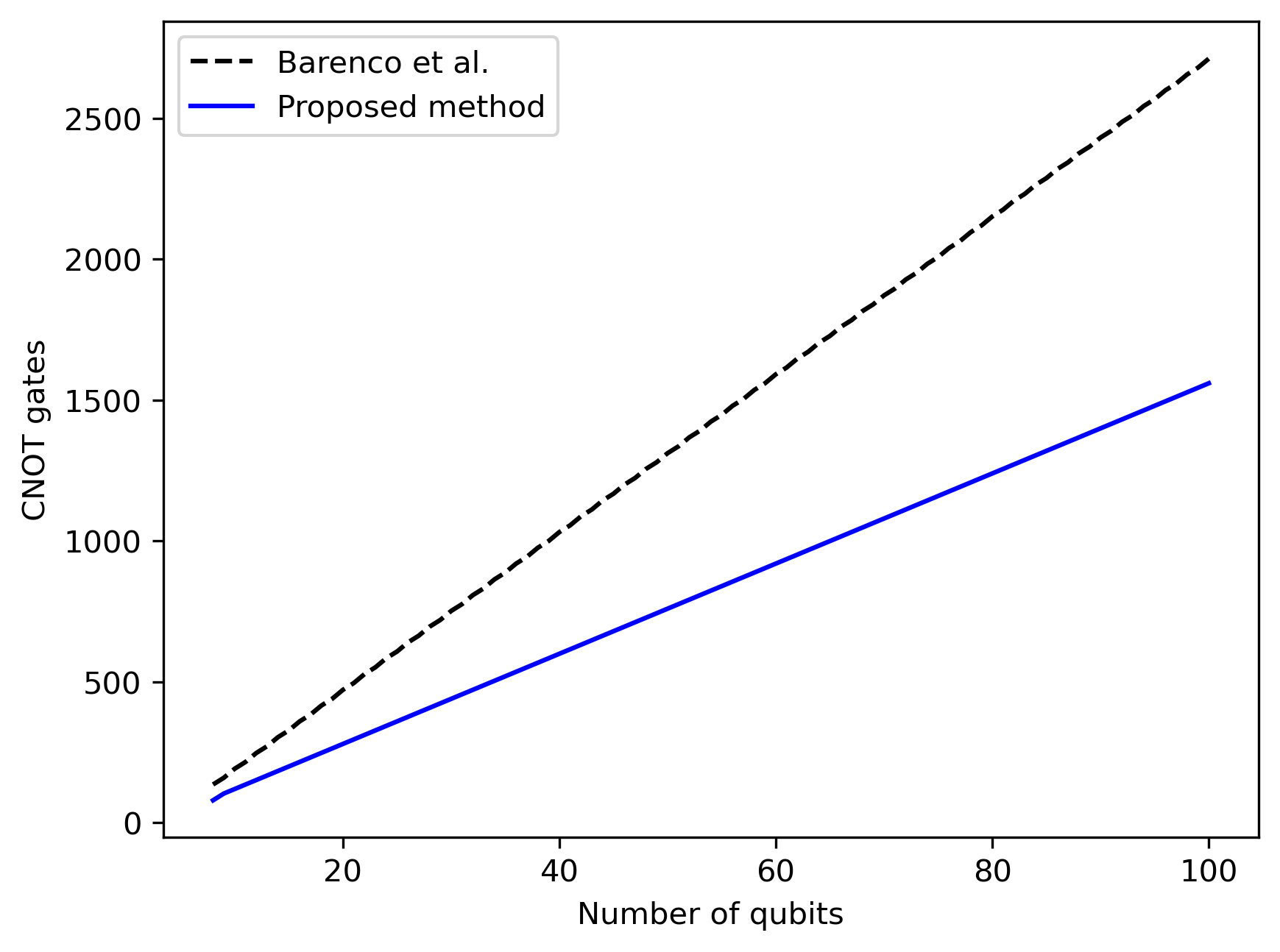}
    \caption{Comparison of the number of CNOTs of the multi-controlled $R_z$ gate using the decomposition scheme of \cite{barenco_1995} and \cite{iten2016quantum} for special unitary gates (dashed black line) and the proposed method (solid blue line).}
    \label{fig:mcrz_cx_count}
\end{figure}

\subsection{Multi-controlled SU(2) gates}

We can also use the eigendecomposition of $\V{}$ and the results of the previous sections to construct a multi-controlled version for any gate $\V{} \in SU(2)$ with $20n-38$ ($20n-42$, $n$ even) CNOTs.

\begin{theorem}
Using Theorem~\ref{main-diag-volta} and eigendecomposition, it is possible to construct a circuit that generates any $n$-controlled $SU(2)$ gate.
\end{theorem}

\begin{proof}
Given $\V{} \in SU(2)$, it has an eigendecomposition $QDQ^{-1}$, in which $D$ is a diagonal matrix and $Q$ is formed from the eigenvectors of $\V{}$. We can choose a suitable phase (if $\Vec{v}$ is an eigenvector, then so is $e^{i \theta} \Vec{v}$) so that the matrix $Q$ only has real elements in its main diagonal. Then, an $n$-controlled version of $\V{}$ can be constructed using the decomposition $QDQ^{-1}$ by applying an $n$-controlled version of each corresponding gate sequentially. Since $D$ is a diagonal matrix, its off-diagonal elements are zeros, so it is possible to use the results from Theorem~\ref{off-diag-volta}. Meanwhile, $Q$ and $Q^{-1}$ have real elements in their main diagonal; consequently, we can use the results of Theorem~\ref{main-diag-volta}. Therefore, it is possible to construct the $n$-controlled $SU(2)$ gate.
\end{proof}

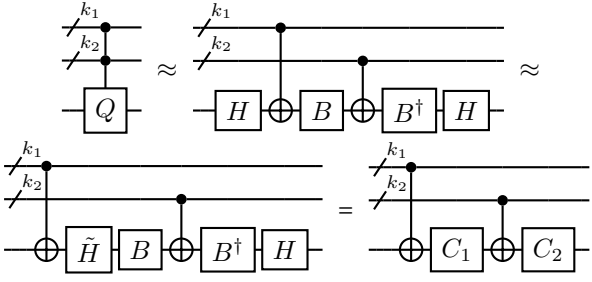
\begin{figure}[ht]
    \centering
    \begin{quantikz}[column sep=0.1cm, row sep=0.3cm]
        & \qw & \qwbundle{k_1} & \ctrl{1} & \qw & \qw \\
        & \qw & \qwbundle{k_2} & \ctrl{1} & \qw & \qw \\
        & \qw & \qw            & \gate{Q} & \qw & \qw
    \end{quantikz}
    ~$\approx$\begin{quantikz}[column sep=0.1cm, row sep=0.3cm]
    & \qw & \qwbundle{k_1} & \qw      & \ctrl{2} & \qw      & \qw      & \qw                & \qw      & \qw & \qw \\
    & \qw & \qwbundle{k_2} & \qw      & \qw      & \qw      & \ctrl{1} & \qw                & \qw      & \qw & \qw \\
    & \qw & \qw            & \gate{H} & \targ{}  & \gate{B} & \targ{}  & \gate{B^{\dagger}} & \gate{H} & \qw & \qw
    \end{quantikz}
    ~$\approx$\begin{quantikz}[column sep=0.1cm, row sep=0.3cm]
    & \qw & \qwbundle{k_1} & \qw & \ctrl{2} & \qw              & \qw      & \qw      & \qw                & \qw      & \qw & \qw \\
    & \qw & \qwbundle{k_2} & \qw & \qw      & \qw              & \qw      & \ctrl{1} & \qw                & \qw      & \qw & \qw \\
    & \qw & \qw            & \qw & \targ{}  & \gate{\tilde{H}} & \gate{B} & \targ{}  & \gate{B^{\dagger}} & \gate{H} & \qw & \qw
    \end{quantikz}
    ~=\begin{quantikz}[column sep=0.1cm, row sep=0.3cm]
    & \qw & \qwbundle{k_1} & \qw & \ctrl{2} & \qw      &\qw       &\qw       & \qw & \qw \\
    & \qw & \qwbundle{k_2} & \qw & \qw      & \qw      & \ctrl{1} &\qw       & \qw & \qw \\
    & \qw & \qw            & \qw & \targ{}  & \gate{C_1} & \targ{}  & \gate{C_2} & \qw & \qw
    \end{quantikz}
    \caption{Decomposition of the multi-controlled gate represented by the eigenvector matrix $Q$. The Hadamard gates can be combined with the adjacent $B$ gates, reducing the total number of operators and exposing the $k_1$-controlled $\sigma_x$. \label{fig:Q}}
\end{figure}

In particular, we can do further optimizations. The usual decomposition needs four operators to guarantee the cancellations that lead to the identity given any configuration of the control qubits, except all active (see Equation~\eqref{eq:circuit_identity_1} and Equation~\eqref{eq:circuit_identity_2}).
But the first and last decomposition blocks of the circuit proposed here are inverses. Thus, we take advantage of this symmetry to guarantee cancellation. Therefore, for $Q$ and $Q^{-1}$ we only need half of the circuit depicted in Fig.~\ref{fig:he_adapted}, and the circuit for $Q$ is shown in Fig.~\ref{fig:Q}, where
\begin{align*}
    B &= \begin{pmatrix}
        \alpha' & -\beta'^*\\
        \beta' & \alpha'^*
    \end{pmatrix} \\
     \alpha' &= \sqrt{\frac{\operatorname{Re}(z)+1}{2}} + i \frac{\operatorname{Im}(z)}{\sqrt{2(\operatorname{Re}(z)+1)}} \\
      \beta' &= \frac{x}{\sqrt{2(\operatorname{Re}(z)+1)}}.
\end{align*}

The circuit for $Q^{-1}$ is the same, but inverted. As for the diagonal gate $D$, we can use the circuit shown in Fig. \ref{fig:he_adapted} since its off-diagonal is real-valued. The three circuits just described are concatenated to form the final circuit, as illustrated in Fig.~\ref{fig:general_su2}.

\begin{figure}[ht]
    \centering
    \resizebox{1.0\columnwidth}{!}
    {
    \begin{quantikz}[column sep=0.1cm, row sep=0.4cm]
    & \qw & \qwbundle{k_1} & \qw & \qw & \qw & \qw & \qw \gategroup[3, steps=4, style={dashed, rounded corners,inner xsep=1pt}, background]{$Q^\dagger$} & \qw     & \qw          & \ctrl{2} & \qw & \ctrl{2} \gategroup[3, steps=8, style={dashed, rounded corners,inner xsep=1pt}, background]{$D$} & \qw & \qw & \qw  & \ctrl{2} & \qw & \qw & \qw & \qw & \qw & \qw & \ctrl{2} \gategroup[3, steps=4, style={dashed, rounded corners,inner xsep=1pt}, background]{$Q$}      & \qw     & \qw     & \qw          & \qw   & \qw & \qw \\
    & \qw & \qwbundle{k_2} & \qw & \qw & \qw & \qw & \qw     & \ctrl{1} & \qw          & \qw      & \qw & \qw & \qw  & \ctrl{1} & \qw & \qw & \qw & \ctrl{1} & \qw & \qw & \qw & \qw & \qw  & \qw   & \ctrl{1} & \qw          & \qw       & \qw & \qw  \\
    & \qw & \qw & \qw & \qw & \qw & \qw & \gate{C_2^\dagger} & \targ{} & \gate{C_1^{\dagger}} &\targ{} & \qw & \targ{} & \gate{A}&\targ{} & \gate{A^\dagger} & \targ{} & \gate{A} & \targ{} & \gate{A^\dagger} & \qw & \qw & \qw & \targ{} & \gate{C_1} & \targ{} & \gate{C_2} & \qw & \qw & \qw 
    \end{quantikz}
    }
    \caption{Decomposition of a multi-controlled $SU(2)$ gate. Each block of the circuit represents one term of the operator's eigendecomposition.}
    \label{fig:general_su2}
\end{figure}
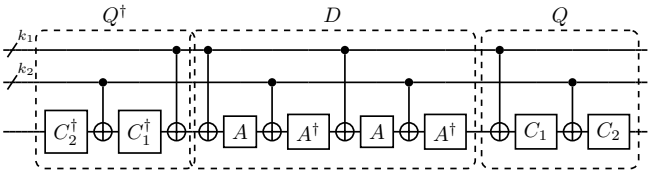

For the $Q$ operator decomposition, we can swap the position of the first Hadamard with the first $k_1$-controlled $\sigma_x$, as shown in Fig.~\ref{fig:Q}, such that $\sigma_x H \equiv \tilde{H} \sigma_x$ for
\begin{equation}
 \tilde{H} =\frac{1}{\sqrt{2}} \begin{pmatrix}
        -1 & 1 \\
        1 & 1
        \end{pmatrix}.
\end{equation}
This procedure allows the cancellation of both the final $k_1$-controlled $\sigma_x$ gate from the $Q^{-1}$ circuit; and the first $k_1$-controlled $\sigma_x$ gate from the $D$ circuit; and the cancellation of the auxiliary qubit reversing parts (see Fig.~\ref{fig:toffoli-chain-decomposition}) of the $k_2$-controlled $\sigma_x$ gate from the $Q^{-1}$ circuit; and the first $k_2$-controlled $\sigma_x$ gate from the $D$ circuit. It also allows the combination of the adjacent gates $B\tilde{H}=C_1$ and $HB^\dagger=C_2$. This last optimization also applies to Theorem~\ref{main-diag-volta}. That way, the total number of operations for both Theorem~\ref{off-diag-volta} and Theorem~\ref{main-diag-volta} is the same.

\subsubsection{Complexity} 

The multi-controlled $\sigma_x$ gates that interact with the top $k_1$ control qubits each need $8 k_1 - 6$ CNOTs \cite[Lemma 8]{iten2016quantum}. Since two of these gates are canceled, one in the $Q^{-1}$ circuit and another in the $D$ circuit, the contribution of the $k_1$-controlled $\sigma_x$ gates is of $16 k_1 - 12$ CNOTs. Due to gate canceling on the first two $\sigma_x$ operators controlled by the bottom $k_2$ controls (one from the $Q^{-1}$ circuit and the other from the $D$ circuit), these gates only perform the target flipping part of their circuits (see Fig.~\ref{fig:toffoli-chain-decomposition}). The target flipping part of one of these gates needs twelve CNOTs to apply two Toffoli gates and up to three CNOTs per application of each approximate Toffoli gate. There are $k_2 - 3$ pairs of approximate Toffoli gates where gate canceling occurs, so each pair contributes with four CNOTs. One approximate Toffoli remains, resulting in a total number of CNOTs equal to $12 + (k_2 - 3)4 + 3 = 4 k_2 + 3$ for the two reduced $k_2$-controlled $\sigma_x$ gates. Finally, there is no additional canceling of gates in the remaining two $k_2$-controlled $\sigma_x$ operators. Therefore, the total cost of the circuit is at most
\begin{equation}
\begin{aligned}
    N_{\textnormal{CNOT}} &= (16 k_1 - 12) + 2(4 k_2 + 3) + (16 k_2 - 12) \\
    &= 16(k_1 + k_2) + 8 k_2 - 18 \\
    &= 16(n - 1) + 8 \Bigl\lfloor \frac{n-1}{2} \Bigr\rfloor - 18 \\
    &= 16n + 8 \Bigl\lfloor \frac{n-1}{2} \Bigr\rfloor - 34
\end{aligned}
\end{equation}
With that, and given $\lfloor \frac{n-1}{2} \rfloor = \frac{n-1}{2}$ for odd $n$ and $\lfloor \frac{n-1}{2} \rfloor = \frac{n}{2} - 1$ for even $n$, we have the following final CNOT count
\begin{equation}
    \begin{cases}
        20n - 38, n~\textnormal{odd} \\
        20n - 42, n~\textnormal{even}
    \end{cases}
    \label{eq:cx_count_su2}
\end{equation}

The results obtained in Equation~\eqref{eq:cx_count_su2} can now be formalized in the following theorem.

\begin{theorem}
    \label{th:cx_count_su2}
    The quantum circuit shown in Fig. \ref{fig:general_su2} can be implemented as an $n$-qubit circuit, where $n \geq 3$, with at most $20n-38$ CNOTs if $n$ is odd or $20n-42$ CNOTs if $n$ is even.
\end{theorem}

\section{Experiments} \label{sec:experiments}

As a use case, we apply the new multi-controlled gate to reduce the cost of circuits produced by the CVO-QRAM sparse state preparation algorithm~\cite{veras2022double}.

The algorithm takes advantage of data storing in quantum random access memory to represent sparse data in the number of patterns stored.
Additionally, the computational cost depends on the number of $1$s in stored patterns, as opposed to the number of qubits.
The circuits produced by the CVO-QRAM technique have an auxiliary qubit beside the memory qubits and begin by initializing the complete register as $\ket{u}\ket{m}^{\otimes n-1}=\ket{1}\ket{0}^{\otimes n-1}$.
For each input vector pattern $p_k$, the multi-controlled $U^{(x_k,\gamma_k)}$ gate, which is defined as
\begin{equation} \label{eq:matrixCU}
    U^{(x_k,\gamma_k)} = \begin{pmatrix}
    \sqrt{\frac{\gamma_{k}-|x_{k}|^2}{\gamma_{k}}}  & \frac{x_{k}}{\sqrt{\gamma_{k}}} \\
    \frac{-x_{k}^{*}}{\sqrt{\gamma_{k}}} & \sqrt{\frac{\gamma_{k}-|x_{k}|^2}{\gamma_{k}}}
    \end{pmatrix},
\end{equation}
encodes the corresponding value $x_k$ as a state amplitude (where $\gamma_{k}=\gamma_{k-1}-|x_{k-1}|^2$ and $\gamma_{0}=1$), plus two CNOT gates are applied before and after the multi-controlled gate, as depicted in Fig.~\ref{fig:cvo_qram_circuit}.

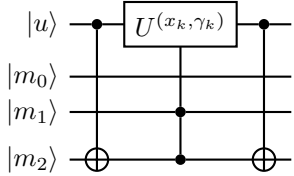
\begin{figure}[ht]
    \centering
    \begin{quantikz}[column sep=0.2cm, row sep=0.4cm]
        \lstick{\ket{u}} & \ctrl{3} & \gate{U^{(x_k,\gamma_k)}} & \ctrl{3} & \qw \\
        \lstick{\ket{m_0}} & \qw & \qw & \qw & \qw \\
        \lstick{\ket{m_1}} & \qw & \ctrl{-2} & \qw & \qw \\
        \lstick{\ket{m_2}} & \targ{} & \ctrl{-1} & \targ{} & \qw
    \end{quantikz}
    \caption{Loading $x_k|011\rangle$ with CVO-QRAM.} \label{fig:cvo_qram_circuit}
\end{figure}

Realizing that the multi-controlled operation is the main cause for the computational cost of the circuit and that the $U^{(x_k,\gamma_k)}$ operator belongs to the $SU(2)$ group with a real-valued main diagonal, the decomposition is readily replaced by the new linear version, reducing the number of CNOTs from $28n-88$ ($28n-92$ for odd $n$)~\cite{barenco_1995, iten2016quantum} to $16n-40$. The advantage of this modification is demonstrated in two experiments.

\begin{figure}[ht]
    \centering
    \includegraphics[width=\linewidth]{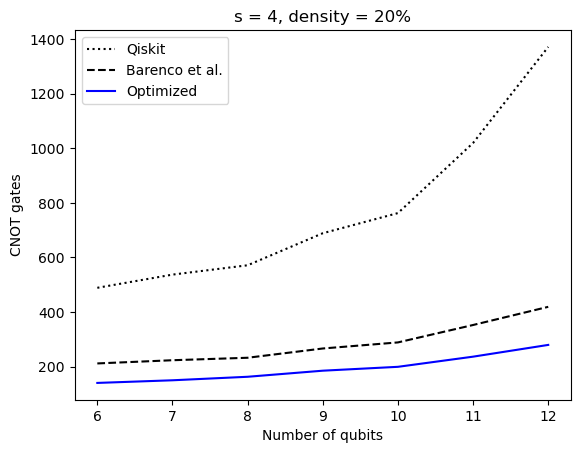}
    \caption{Average number of CNOT gates produced by CVO-QRAM algorithm for a double sparse random state with $n$ qubits and $2^s$ nonzero entries. The density is the average number of 1s in the binary strings.}
    \label{fig:cvo_qram}
\end{figure}

The first experiment, shown in Fig.~\ref{fig:cvo_qram}, compares the number of CNOTs on circuits produced by CVO-QRAM (using Qiskit's multi-controlled gate~\cite{qiskit} and the method presented by Barenco et al. in~\cite{barenco_1995} with the improvements from~\cite{iten2016quantum}) and by the optimized CVO-QRAM for double sparse states with the number of qubits ranging from $n=6$ to $n=12$, with $2^4$ nonzero amplitudes, and a 20\% average density of $1$s present in the binary strings. Each point on the graph is an average of 30 different random state results. Fig. \ref{fig:cvo_qram} shows that circuits produced by CVO-QRAM have significantly more CNOTs than the ones by its optimized version.
This experiment does not target any device and does not use Qiskit's circuit optimization.

\begin{figure}[ht]
    \centering
    \begin{subfigure}[b]{.4\linewidth}
        \includegraphics[width=\linewidth]{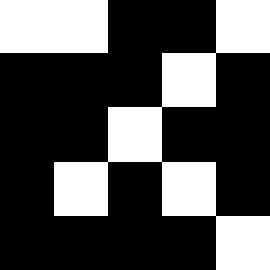}
        \caption{Ideal}
    \end{subfigure}
    \begin{subfigure}[b]{.4\linewidth}
        \includegraphics[width=\linewidth]{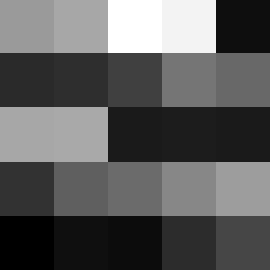}
        \caption{Qiskit}
    \end{subfigure} \\ \vspace{0.5cm}
    \begin{subfigure}[b]{.4\linewidth}
        \includegraphics[width=\linewidth]{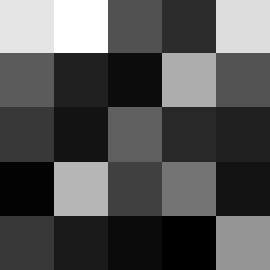}
        \caption{Barenco et al.}
    \end{subfigure}
    \begin{subfigure}[b]{.4\linewidth}
        \includegraphics[width=\linewidth]{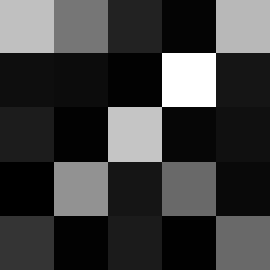}
        \caption{Optimized}
    \end{subfigure}
    \caption{Visual comparison of a $5$-qubit double sparse quantum state preparation. (a) Visual representation of the ideal measurement probabilities. (b) State initialized by nonoptimized CVO-QRAM using Qiskit's multi-controlled gate. (c) State initialized by CVO-QRAM using the multi-controlled gate proposed in~\cite{barenco_1995}. (d) State initialized by CVO-QRAM using the multi-controlled gate proposed in this work. These illustrations represent an estimate of the measurement probabilities on the ibm\_oslo device.}
    \label{fig:compare_cvo_qram}
\end{figure}

The second experiment, depicted in Fig.~\ref{fig:compare_cvo_qram} and summarized in Table~\ref{tab:circuit_cvo}, compares the performance of CVO-QRAM on IBM's ibm\_oslo quantum device using a different implementation of the multi-controlled gate. This experiment prepares a $5$-qubit double sparse state and estimates the measurement probabilities. The algorithm performance is evaluated by comparing the estimate against ideal values. Table~\ref{tab:circuit_cvo} shows the mean absolute error (MAE) between the estimated and the ideal probabilities. The MAE produced with optimized CVO-QRAM is smaller than that of the nonoptimized ones. Fig.~\ref{fig:compare_cvo_qram} is a visual representation of this result. This experiment uses Qiskit's circuit optimization level 3, and each figure is produced from one execution on the target device, with 8192 shots.

\begin{table}[ht]
	\centering
	\begin{tabular}{lccc}
		\hline
		& CNOTs & Depth & MAE \\
    	\hline
		  Qiskit & $149$ & $320$ & $0.04541$ \\
            Barenco et al. & $82$ & $190$ & $0.03137$ \\
		Optimized & $39$ & $113$ & $0.01654$ \\
		\hline
	\end{tabular}
	\caption{Number of CNOTs and depth of the circuit produced by the CVO-QRAM algorithms to encode a $5$-qubit double sparse state with eight nonzero amplitudes with an average of 10\% of $1$s in the binary strings using multi-controlled gates based on Qiskit~\cite{qiskit}, Barenco et al.~\cite{barenco_1995}, and the optimized method proposed in this work. The MAE column shows the mean absolute error between the ideal and the actual results.}
	\label{tab:circuit_cvo}
\end{table}

All experiments were performed using the \textit{Qclib} library~\cite{qclib}.

\section{Conclusion} \label{sec:conclusion}

In this paper, we proposed a linear decomposition for $n$-qubit multi-controlled special unitary single-qubit gates without auxiliary qubits.
Our method shows improved gate counts and depth over the best schemes to decompose general multi-controlled $SU(2)$ gates \cite{barenco_1995,iten2016quantum} known so far, with $20n - 38$ CNOTs ($20n - 42$ for even $n$) needed to use our construction compared with $28n - 88$ CNOTs ($28n - 92$ for odd $n$).
We have also presented an additional scheme for $SU(2)$ gates with matrices containing at least one real-valued diagonal, which yields an improved CNOT count of $16n - 40$. \cite{veras2022double} is suggested as a possible method in such cases, which, as we have shown, can have a total number of CNOT gates lower than originally estimated.

Some future considerations are the theoretical bounds when constructing the types of gates described in this paper and whether the proposed decomposition scheme achieves or approaches optimality.
It is also worth noting that we have not developed similar methods with the inclusion of auxiliary qubits; thus, any potential enhancements that result from doing so are yet to be explored.
We would also like to investigate the generalization of a decomposition method for any multi-controlled $U(2)$ gate that aims to maintain a lower CNOT cost and depth.
In particular, a significant challenge to circumvent is the introduction of numerical error in the computation of gates generated in the decomposition of $U(2)$ gates by some known methods \cite{barenco_1995,adenilton2022linear}.
As a result, for a decomposition scheme for these types of gates to be viable in practice, simply the reduction in depth and number of gates as a goal is not sufficient, pointing to the importance of different solutions to bypass this issue.

\section*{Acknowledgments}

This work is based upon research supported by CNPq (Grant No. 409506/2022-2, No. 409513/2022-9 and No. 162052/2021-9), CAPES -- Finance Code 001, CAPES (Grant No. 25001019004P6), FACEPE (Grant No. APQ-1229-1.03/21), National Research Foundation of Korea (Grant No. 2022M3E4A1074591).
We acknowledge the use of IBM Quantum services for this work. The views expressed are those of the authors, and do not reflect the official policy or position of IBM or the IBM Quantum team.

\section*{Data availability}

The sites \url{https://github.com/qclib/qclib-papers} and \url{https://github.com/qclib/qclib} contain all the data and the software generated during the current study.

\begin{appendices}

\section{Theorem 1 details}
\label{apd:theorem1_details}

In this section, we detail the mathematical steps taken to prove Theorem~\ref{off-diag-volta}. First, we start with the equations for $z$
\begin{equation}\label{eq:: sys_initial_z}
    z = \omega_1^2 - \omega_2^2
\end{equation}
and $x$
\begin{equation}\label{eq:: sys_initial_x}
    x = 2 \operatorname{Re}(\omega_1) \omega_2,
\end{equation}
as well as the requirement of unitarity for both $V$ from Equation~\eqref{eq:matrix_V} and the matrix defined in Equation~\eqref{eq::aux_w}:
\begin{equation}\label{eq::unitarity}
    \begin{cases}
        \vert z \vert^2 + x^2 =1\\
         \vert \omega_1 
         \vert^2+\omega_2^2 = 1
    \end{cases}
\end{equation}

Expanding Equation~\eqref{eq:: sys_initial_z}, we obtain
\begin{multline}
    \operatorname{Re}(z) + i \operatorname{Im}(z) = \\ \operatorname{Re}(\omega_1)^2 - \operatorname{Im}(\omega_1)^2-\omega_2^2 +2i \operatorname{Re}(\omega_1)\operatorname{Im}(\omega_1).
\end{multline}
Looking at the real elements and using the unitarity from Equation~\eqref{eq::unitarity} we have
\begin{equation}
    \operatorname{Re}(z) = 2\operatorname{Re}(\omega_1)^2-1,
\end{equation}
which leads to
\begin{equation}\label{eq::re_omega_1}
    \operatorname{Re}(\omega_1) = \pm \sqrt{\dfrac{\operatorname{Re}(z)+1}{2}}\rightarrow \sqrt{\dfrac{\operatorname{Re}(z)+1}{2}},
\end{equation}
in which we have chosen the positive solution. Replacing Equation~\eqref{eq::re_omega_1} into Equation~\eqref{eq:: sys_initial_z} and Equation~\eqref{eq:: sys_initial_x}:

\begin{equation}
    \operatorname{Im}(\omega_1) = \frac{\operatorname{Im}(z)}{\sqrt{2(\operatorname{Re}(z)+1)}}
\end{equation}

\begin{equation}
    \omega_2 = \frac{x}{\sqrt{2(\operatorname{Re}(z)+1)}}
\end{equation}

Now we proceed to determine $\alpha$ and $\beta$. First, we can write $\alpha=a+bi$, $\beta=c+di$. So, 

\begin{equation}\label{eq::omega1}
        \omega_1=(a^2-b^2-c^2+d^2)+2(ab+cd)i\\
\end{equation}

\begin{equation}
     \omega_2=2(ac+bd)
\end{equation}

Now, $\omega_1$ and $\omega_2$ have three free variables, which can be reduced to two free variables due to unitarity. Meanwhile, $\alpha$ and $\beta$ have three free variables after accounting for unitarity. With the extra free variable, we make a choice for $d=0$, making $\beta$ strictly real, which gives us
\begin{equation} \label{eq::c_omega}
    c=\frac{\omega_2}{2a}
\end{equation}
and
\begin{equation}\label{eq::b_omega}
    b = \frac{\operatorname{Im}(\omega_1)}{2a}.
\end{equation}

Plugging in these results in Equation~\eqref{eq::omega1} gives us a degree 4 equation:
\begin{equation}
    a^4-\operatorname{Re}(\omega_1)a^2+\dfrac{\operatorname{Re}(\omega_1)^2-1}{4}=1
\end{equation}
We choose a real and positive solution for $a$:
\begin{equation}
    a = \sqrt{\frac{\operatorname{Re}(\omega_1)+}{2}} = \sqrt{\frac{\sqrt{\frac{\operatorname{Re}(z)+1}{2}} + 1}{2}}
\end{equation}

One can verify that $|a|^2+|b|^2+|c|^2=1$. Replacing $a$ into Equation~\eqref{eq::c_omega} and Equation~\eqref{eq::b_omega} gives us
\begin{equation}
   b =  \frac{\operatorname{Im}(z)}{2 \sqrt{(\operatorname{Re}(z)+1) \left( \sqrt{\frac{\operatorname{Re}(z)+1}{2}} + 1 \right) }}
\end{equation}
and
\begin{equation}
    c = \frac{x}{2 \sqrt{(\operatorname{Re}(z)+1) \left( \sqrt{\frac{\operatorname{Re}(z)+1}{2}} + 1 \right)}},
\end{equation}
which is the result from Theorem~\ref{off-diag-volta}.

\end{appendices}



\begin{thebibliography}{10}
\bibitem{shor1999polynomial}
P.~W. Shor, ``Polynomial-time algorithms for prime factorization and discrete
  logarithms on a quantum computer,'' \emph{SIAM review}, vol.~41, no.~2, pp.
  303--332, 1999.

\bibitem{grover1997quantum}
L.~K. Grover, ``Quantum mechanics helps in searching for a needle in a
  haystack,'' \emph{Physical review letters}, vol.~79, no.~2, p. 325, 1997.

\bibitem{preskill2018quantum}
J.~Preskill, ``Quantum computing in the {NISQ} era and beyond,''
  \emph{Quantum}, vol.~2, p.~79, 2018.

\bibitem{takagi_fundamental_2022}
R.~Takagi, S.~Endo, S.~Minagawa, and M.~Gu,
  ``Fundamental limits of quantum error mitigation,''
  \emph{npj Quantum Information}, vol.~8, no.~1, p.
  114, 2022.

\bibitem{boixo2018characterizing}
S.~Boixo, S.~V. Isakov, V.~N. Smelyanskiy, R.~Babbush, N.~Ding, Z.~Jiang, M.~J.
  Bremner, J.~M. Martinis, and H.~Neven, ``Characterizing quantum supremacy in
  near-term devices,'' \emph{Nature Physics}, vol.~14, no.~6, pp. 595--600,
  2018.

\bibitem{brugiere2021reducing}
T.~G. De~Brugiere, M.~Baboulin, B.~Valiron, S.~Martiel, and C.~Allouche,
  ``Reducing the depth of linear reversible quantum circuits,'' \emph{IEEE
  Transactions on Quantum Engineering}, vol.~2, pp. 1--22, 2021.

\bibitem{nguyen2021enabling}
T.~Nguyen and A.~McCaskey, ``Enabling pulse-level programming, compilation, and
  execution in {XACC},'' \emph{IEEE Transactions on Computers}, vol.~71, no.~3,
  pp. 547--558, 2021.

\bibitem{he2017decompositions}
Y.~He, M.-X. Luo, E.~Zhang, H.-K. Wang, and X.-F. Wang, ``Decompositions of
  n-qubit {T}offoli gates with linear circuit complexity,'' \emph{International
  Journal of Theoretical Physics}, vol.~56, no.~7, pp. 2350--2361, 2017.

\bibitem{araujo2021approximated}
I.~F. Araujo, C.~Blank, I.~Cesar, and A.~J. da~Silva, ``Approximated
  quantum-state preparation with entanglement dependent complexity,''
  \emph{arXiv preprint arXiv:2111.03132}, 2021.

\bibitem{barenco_1995}
A.~Barenco, C.~H. Bennett, R.~Cleve, D.~P. DiVincenzo, N.~Margolus, P.~Shor,
  T.~Sleator, J.~A. Smolin, and H.~Weinfurter, ``Elementary gates for quantum
  computation,'' \emph{Physical Review A}, vol.~52, pp. 3457--3467, 1995.

\bibitem{iten2016quantum}
R.~Iten, R.~Colbeck, I.~Kukuljan, J.~Home, and M.~Christandl, ``Quantum
  circuits for isometries,'' \emph{Physical Review A}, vol.~93, no.~3, p.
  032318, 2016.

\bibitem{adenilton2022linear}
A.~J. da~Silva and D.~K. Park, ``Linear-depth quantum circuits for multiqubit
  controlled gates,'' \emph{Physical Review A}, vol. 106, p. 042602, 2022.

\bibitem{saeedi2013linear}
M.~Saeedi and M.~Pedram, ``Linear-depth quantum circuits for n-qubit toffoli
  gates with no ancilla,'' \emph{Physical Review A}, vol.~87, no.~6, p. 062318,
  2013.

\bibitem{gokhale2019qutrits}
P.~Gokhale, J.~M. Baker, C.~Duckering, N.~C. Brown, K.~R. Brown, and F.~T.
  Chong, ``Asymptotic improvements to quantum circuits via qutrits,'' in
  \emph{Proceedings of the 46th International Symposium on Computer
  Architecture}, 2019, pp. 554--566.

\bibitem{biswal2019fault}
L.~Biswal, D.~Bhattacharjee, A.~Chattopadhyay, and H.~Rahaman, ``Techniques for
  fault-tolerant decomposition of a multicontrolled toffoli gate,''
  \emph{Physical Review A}, vol. 100, no.~6, p. 062326, 2019.

\bibitem{balauca2022efficient}
S.~Balauca and A.~Arusoaie, ``Efficient constructions for simulating multi
  controlled quantum gates,'' in \emph{International Conference on
  Computational Science}.\hskip 1em plus 0.5em minus 0.4em\relax Springer,
  2022, pp. 179--194.

\bibitem{tomesh2022variational}
T.~Tomesh, N.~Allen, and Z.~Saleem, ``Quantum-classical tradeoffs and
  multi-controlled quantum gate decompositions in variational algorithms,''
  \emph{arXiv preprint arXiv:2210.04378}, 2022.

\bibitem{kim2018ancillary}
T.~Kim and B.-S. Choi, ``Efficient decomposition methods for controlled-{R}n
  using a single ancillary qubit,'' \emph{Scientific reports}, vol.~8, no.~1,
  pp. 1--7, 2018.

\bibitem{lloyd2014quantum}
S.~Lloyd, M.~Mohseni, and P.~Rebentrost, ``Quantum principal component
  analysis,'' \emph{Nature Physics}, vol.~10, no.~9, pp. 631--633, 2014.

\bibitem{harrow2009quantum}
A.~W. Harrow, A.~Hassidim, and S.~Lloyd, ``Quantum algorithm for linear systems
  of equations,'' \emph{Physical review letters}, vol. 103, no.~15, p. 150502,
  2009.

\bibitem{li2023quantum}
J.~Li, F.~Gao, S.~Lin, M.~Guo, Y.~Li, H.~Liu, S.~Qin, and Q.~Wen, ``Quantum
  k-fold cross-validation for nearest neighbor classification algorithm,''
  \emph{Physica A: Statistical Mechanics and its Applications}, vol. 611, p.
  128435, 2023.

\bibitem{guo2022quantum}
M.~Guo, H.~Liu, Y.~Li, W.~Li, F.~Gao, S.~Qin, and Q.~Wen, ``Quantum algorithms
  for anomaly detection using amplitude estimation,'' \emph{Physica A:
  Statistical Mechanics and its Applications}, vol. 604, p. 127936, 2022.

\bibitem{schuld2016prediction}
M.~Schuld, I.~Sinayskiy, and F.~Petruccione, ``Prediction by linear regression
  on a quantum computer,'' \emph{Physical Review A}, vol.~94, no.~2, p. 022342,
  2016.

\bibitem{biamonte2017quantum}
J.~Biamonte, P.~Wittek, N.~Pancotti, P.~Rebentrost, N.~Wiebe, and S.~Lloyd,
  ``Quantum machine learning,'' \emph{Nature}, vol. 549, no. 7671, pp.
  195--202, 2017.

\bibitem{park2019circuit}
D.~K. Park, F.~Petruccione, and J.-K.~K. Rhee, ``Circuit-based quantum random
  access memory for classical data,'' \emph{Scientific reports}, vol.~9, no.~1,
  p. 3949, 2019.

\bibitem{gleinig2021efficient}
N.~Gleinig and T.~Hoefler, ``An efficient algorithm for sparse quantum state
  preparation,'' in \emph{2021 58th ACM/IEEE Design Automation Conference
  (DAC)}.\hskip 1em plus 0.5em minus 0.4em\relax IEEE, 2021, pp. 433--438.

\bibitem{low2014quantum}
G.~H. Low, T.~J. Yoder, and I.~L. Chuang, ``Quantum inference on bayesian
  networks,'' \emph{Physical Review A}, vol.~89, no.~6, p. 062315, 2014.

\bibitem{orus2019quantum}
R.~Or{\'u}s, S.~Mugel, and E.~Lizaso, ``Quantum computing for finance: Overview
  and prospects,'' \emph{Reviews in Physics}, vol.~4, p. 100028, 2019.

\bibitem{mozafari2021efficient}
F.~Mozafari, H.~Riener, M.~Soeken, and G.~De~Micheli, ``Efficient boolean
  methods for preparing uniform quantum states,'' \emph{IEEE Transactions on
  Quantum Engineering}, vol.~2, pp. 1--12, 2021.

\bibitem{de2021classical}
L.~S. de~Souza, J.~H. de~Carvalho, and T.~A. Ferreira, ``Classical artificial
  neural network training using quantum walks as a search procedure,''
  \emph{IEEE Transactions on Computers}, vol.~71, no.~2, pp. 378--389, 2021.

\bibitem{de2020circuit}
T.~M. De~Veras, I.~C. De~Araujo, D.~K. Park, and A.~J. Da~Silva,
  ``Circuit-based quantum random access memory for classical data with
  continuous amplitudes,'' \emph{IEEE Transactions on Computers}, vol.~70,
  no.~12, pp. 2125--2135, 2020.

\bibitem{qiskit}
G.~Aleksandrowicz and et~al., ``Qiskit: An open-source framework for quantum
  computing,'' 2021.

\bibitem{veras2022double}
T.~M. de~Veras, L.~D. da~Silva, and A.~J. da~Silva, ``Double sparse quantum
  state preparation,'' \emph{Quantum Information Processing}, vol.~21, no.~6,
  pp. 1--13, 2022.

\bibitem{qclib}
I.~F. Araujo, I.~C.~S. Araújo, L.~D. da~Silva, C.~Blank, and A.~J. da~Silva,
  ``{Quantum computing library},'' 7 2022. [Online]. Available:
  \url{https://github.com/qclib/qclib}

\end{thebibliography}
\end{document}